\documentclass[twocolumn]{autart}

\usepackage{graphicx}
\usepackage{amsmath,amssymb,mathtools}

\usepackage{cite}
\usepackage{tikz}
\usepackage{url}

\newenvironment{proof}{\begin{pf}}{\end{pf}}

\newtheorem{proposition}{Proposition}[section]
\newtheorem{lemma}{Lemma}[section]
\newtheorem{theorem}{Theorem}[section]
\newtheorem{definition}{Definition}[section]
\newtheorem{remark}{Remark}[section]
\newtheorem{assumption}{Assumption}[section]
\DeclareMathOperator{\col}{col}
\newcommand{\R}{\mathbb{R}}
\newcommand{\E}{\mathbb{E}}
\newcommand{\1}{\mathbf{1}}
\newcommand{\diag}{\operatorname{diag}}
\newcommand{\norm}[1]{\left\lVert #1 \right\rVert}

\begin{document}

\begin{frontmatter}

\title{Delay-Robust Secondary Frequency Control via Passive Interconnection and Randomized Block Updates\thanksref{fund}}

\thanks[fund]{This work was supported in part by the National Natural Science Foundation of China under Grant 52307145, and in part by the Shenzhen Research Institute of Big Data (SRIBD) under Grant J00220250001.}

\author[aff1]{Yiwei Liu}
\author[aff2]{Luwei Yang\corauthref{cor}}
\corauth[cor]{Corresponding author: Luwei Yang.}
\author[aff1,aff2]{Shunbo Lei}

\address[aff1]{The Chinese University of Hong Kong, Shenzhen, China}
\address[aff2]{Shenzhen Research Institute of Big Data, China}

\begin{abstract}
This paper studies secondary frequency control in transmission networks subject to communication delays at the cyber-physical interface and limited
per-update computation at the control center. The regulation objective is formulated as a constrained economic dispatch problem incorporating generation capacity constraints, nodal power balance, transmission-flow limits, and scheduled tie-line power exchanges. Based on this formulation, we develop a passivity-based control framework in which an augmented projected primal-dual controller restores nominal frequency and drives the closed-loop system to the solution set of the constrained economic dispatch problem. Two-way communication delays between
the physical network and the control center are modeled as scattering-based
passive channels for the measurement uplink and the control-command downlink. This construction preserves the target equilibrium and enables a delay-robust passivity analysis of the delayed closed loop. To reduce the computational burden at the control center, we develop a randomized block-coordinate implementation of the augmented projected primal-dual controller. The resulting sampled-data closed loop preserves the target
solution set and achieves local mean-square geometric convergence under
suitable step-size and regularity conditions. Finally, a multivariable wave-domain interface filter is introduced to inject additional
dissipation and improve the damping of the delayed interface without altering the steady-state interconnection. Simulations on the IEEE 14-bus system indicate that the proposed digital implementation accurately reproduces the delayed closed-loop behavior while reducing the per-update computational cost.
\end{abstract}


\end{frontmatter}

\section{Introduction}
\label{sec:introduction}

Secondary frequency control is a fundamental ancillary-service mechanism in power systems and aims to restore the system frequency to its nominal value after supply-demand imbalances while economically allocating regulation effort  and respecting operational constraints \cite{andersonPowerSystemControl2002}. This task has become more challenging with the increasing penetration of renewable generation, whose variability leads to larger and more frequent active-power fluctuations \cite{malladaOptimalLoadSideControl2017}. Hence, there is a growing need for frequency control schemes that combine reliable dynamic performance, systematic constraint enforcement, and real-time computational tractability.

From a control and optimization perspective, secondary frequency regulation in recent years has been typically formulated as a dynamic realization of a constrained~economic dispatch problem interconnected with the physical network dynamics \cite{andersonPowerSystemControl2002}.  Existing control architectures can be broadly categorized as distributed, centralized, and hybrid \cite{li2015connecting,zhang2013real,xi2018power,weitenbergExponentialConvergenceDistributed2018,wangDistributedFrequencyControl2019,yang2020distributed,DORFLER2017296}. Distributed schemes are attractive due to their scalability and reduced reliance on a central coordinator. However, when system-wide constraints such as scheduled tie-line power exchanges and transmission-flow limits are imposed, they often require substantial coordination and may become sensitive to communication delays, asynchronous updates, and network imperfections. Hybrid architectures can mitigate some of these difficulties, but at the cost of additional implementation complexity \cite{DORFLER2017296}. By contrast, centralized architectures are naturally suited to global coordination and the direct enforcement of network-wide constraints, and are~consistent with the energy management infrastructures currently prevalent in practical power systems \cite{li2015connecting,zhang2013real,xi2018power}.

Under a centralized architecture, the control center interacts with the physical grid through measurement uplink and control-command downlink channels, thereby forming a closed-loop cyber-physical system. A classical realization of this paradigm is automatic generation control (AGC), which can be interpreted as a real-time dispatch mechanism that implicitly solves a simplified optimization problem primarily aimed at restoring global power balance \cite{hauswirthOptimizationAlgorithmsRobust2024a}. Although AGC has been highly successful in practice, it does not provide a systematic way to enforce richer operational requirements, such as generation capacity constraints, scheduled tie-line exchanges, and transmission congestion limits.

Primal-dual methods provide a natural framework~for~incorporating the above constraints into frequency control laws \cite{malladaOptimalLoadSideControl2017,zhaoDesignStabilityLoadSide2014}. Under ideal communication assumptions, primal-dual dynamics can be designed so that the closed-loop equilibria satisfy the optimality conditions of the underlying dispatch problem and the resulting trajectories converge to the desired operating points. In practical centralized implementations, however, the measurement uplink and control-command downlink are subject to non-negligible communication delays \cite{yu2004lmi,jiang2012delay,zhang2015measurement}. These delays are not merely implementation artifacts, rather they enter the feedback path between the optimizer and the physical network, and may destroy the structural properties on which delay-free primal-dual designs rely. A delay-aware closed-loop design is therefore required.

A further difficulty arises from real-time computation. In large-scale networks, a full primal-dual update involves high-dimensional variables and network matrices, which can impose a substantial real-time computational burden on the control center. Randomized block-coordinate (RBC) updates offer a possible remedy, since only a subset of the controller  variables is updated at each sampling instant \cite{Nesterov2012EfficiencyCD,RichtarikTakac2014IterationComplexityRBCD,LiuWright2015AsynchronousSCD,ZhangXiao2017SPDC}. Existing results on randomized coordinate and asynchronous methods, however, are primarily developed for static large-scale optimization problems, where the focus is on iteration complexity and computational efficiency. In secondary frequency control, the optimization dynamics are embedded in a feedback loop with the physical power network. It is therefore necessary to understand whether randomized partial updates remain compatible with closed-loop stability and convergence, especially in the presence of cyber–physical interface delays.

Although several closely related research directions have been pursued, to the best of our knowledge, the above
issues have not been systematically addressed in a unified framework. First, delay-robust analyses of primal-dual dynamics have been developed for constrained~optimization \cite{wang2003globaldelaypdcc,yuan2016regularizedpdsubgradient,hale2017asyncmultiagentpd,millar2017smartgridasyncpd,hendrickson2023totallyasyncpd,su2025pdfixedpointdelays,sen2026delayrobustpdgd}. Most of these results, however, are formulated at the
optimization-algorithm level and are often tailored to relatively simple
constraint structures, such as equality constraints or inequality constraints
considered in isolation. This limits their direct applicability to real-time frequency control, where nodal power-balance constraints, transmission-flow limits, scheduled tie-line constraints,  and actuator constraints must be enforced simultaneously within a closed-loop cyber-physical system. Second, feedback optimization provides a general methodology for interconnecting a stable plant with an optimization layer \cite{hauswirthOptimizationAlgorithmsRobust2024a,Hauswirth2021TimescaleSeparation,He2024ModelFreeNonlinearFO}. Nevertheless, feedback-optimization methods typically
rely on sensitivity information associated with the plant steady-state map. In power-system applications, such sensitivities may be difficult to attain accurately, and data-driven approximations may introduce steady-state optimality errors. Moreover, existing feedback-optimization approaches rarely address communication delays across the cyber-physical interface, making their direct use in delayed frequency-regulation architectures nontrivial. Third, stochastic-gradient and coordinate-based methods offer an effective route to improving computational efficiency in large-scale optimization \cite{Nesterov2012EfficiencyCD,ZhangXiao2017SPDC,LatafatFrerisPatrinos2019TriPD}. Their convergence and complexity properties have been extensively studied for algorithms evolving over prescribed
information structures. These results, however, do not account for the dynamics of the underlying power network and therefore cannot be directly used to justify stochastic or coordinate-reduced implementations of delay-robust cyber-physical frequency control.

Motivated by these observations, this paper develops a delay-aware framework for secondary frequency control with a computationally efficient RBC implementation. The main contributions are summarized below. First, we
propose a centralized augmented projected primal–dual controller for frequency regulation under heterogeneous operational constraints, including generation capacity constraints, scheduled tie-line exchange constraints, and transmission congestion constraints, thus yielding a unified closed-loop treatment of constraints that are often handled only partially in existing designs. Second, the delayed controller-plant interconnection is reformulated through passive communication channels. This representation explicitly captures both measurement and actuation delays, and yields a tractable passivity structure for delay-robust
 analysis. Third, a sampled-data RBC realization of the proposed controller is developed to reduce the per-step computational burden at the central optimizer. Under suitable regularity and step-size conditions, the system preserves the desired steady-state optimality properties and achieves local mean-square geometric convergence. Finally, we introduce a multivariable wave-domain
interface filter as an interface-level refinement. The filter adds dissipation
to the delayed interconnection without changing its steady-state input-output
relation. To the best of our knowledge, this is among the first results that provide a rigorous stability analysis for randomized optimizer
updates embedded in a delayed cyber-physical frequency-control loop.

The rest of this paper is organized as follows. Section~\ref{sec:notations} gives the notations and preliminaries. Section~\ref{sec:problem} formulates the problem. The projected primal-dual controller and delay-aware passivity analysis are presented in Section~\ref{sec:controller}. Section~\ref{sec:rbc} gives the RBC implementation and its convergence analysis. Section~\ref{sec:filtering} introduces the passivity-preserving wave-domain interface filtering. Section~\ref{sec:verification} reports simulation results. Section~\ref{sec:conclusion} concludes the paper.

\section{Notations and preliminaries}
\label{sec:notations}

\subsection{Mathematical notation}

Let \(\mathbb{R}\), \(\mathbb{R}_{\geq 0}\) denote the sets of real and nonnegative real numbers, respectively. For positive integers \(m,n\), \(\mathbb{R}^n\) denotes the set of \(n\)-dimensional real vectors, and \(\mathbb{R}^{m\times n}\) denotes the set of \(m\times n\) real matrices. For a vector \(x\in\mathbb{R}^n\), \(x_i\) denotes its \(i\)th entry. The transpose is denoted by \((\cdot)^\top\). The operators \(\diag(\cdot)\) and \(\col(\cdot)\) denote diagonalization and column stacking, respectively. The all-one vector in \(\mathbb{R}^n\) is denoted by \(\mathbf{1}_n\), the identity matrix in \(\mathbb{R}^{n\times n}\) is denoted by \(I_n\), and \(\mathbf{1}_n^\perp := \{x\in\mathbb{R}^n \mid \mathbf{1}_n^\top x = 0\}\) denotes its orthogonal complement. For a set \(\mathcal S\), \(|\mathcal S|\) denotes its cardinality, and \(\operatorname{cl}(\mathcal S)\) denotes its closure, i.e., the smallest closed set containing \(\mathcal S\). The Euclidean norm and inner product are denoted by \(\|\cdot\|\) and \(\langle \cdot,\cdot\rangle\), respectively.  For a symmetric matrix \(M\), the notation \(M\succ 0\) means that \(M\) is positive definite. Unless otherwise specified, vector equalities and inequalities are understood componentwise.

Let \(\Omega\subseteq\mathbb{R}^n\) be a nonempty closed convex set. For any \(x\in\Omega\), the normal cone and tangent cone of \(\Omega\) at \(x\) are defined as \(\mathcal{C}_{\Omega}(x) := \{y\in\mathbb{R}^n \mid \langle y,z-x\rangle \le 0,\ \forall z\in\Omega\}\), and \(T_\Omega(x) := \operatorname{cl}\{\alpha(z-x)\mid \alpha\ge 0,\ z\in\Omega\}\), respectively. For any \(y\in\mathbb{R}^n\), the Euclidean projection of \(y\) onto \(\Omega\) is denoted by \(\mathcal P_{\Omega}(y) := \arg\min_{z\in\Omega}\|z-y\|\). For any \(x\in\Omega\) and any \(y\in\mathbb{R}^n\), the projected velocity is defined by \(\Pi_\Omega(x,y) := \mathcal P_{T_\Omega(x)}(y)\). 

\subsection{Preliminaries}

The controller to be designed uses projection operators to enforce inequality constraints. We therefore recall two standard facts on projected dynamical systems, which
will be used in the subsequent theoretical analysis. The first lemma characterizes projected equilibria through the normal cone 
\cite{ZhangWeiYiHu2018ProjectedPDAL,BauschkeCombettes2017}.

\begin{lemma}[Projected equilibrium]
\label{lem:proj-equilibrium}
Let \(\Omega\subset\mathbb{R}^n\) be a nonempty closed convex set. Then, for any
\(x\in\Omega\) and \(v\in\mathbb{R}^n\), \(\Pi_\Omega(x,v)=0\) if and only
if \(v\in\mathcal C_\Omega(x)\). Equivalently,
\(\langle v,z-x\rangle\le 0\) for all \(z\in\Omega\).
\end{lemma}

The second lemma give the nonexpansiveness-type property of projected velocities, which is a projection inequality that follows
from the variational characterization of Euclidean projections
\cite{BauschkeCombettes2017}.

\begin{lemma}[Projection inequality]
\label{lem:proj-nonexpansive}
Let \(\Omega\subset\mathbb{R}^n\) be a nonempty closed convex set. Then, for
any \(x,z\in\Omega\) and \(\xi\in\mathbb{R}^n\), \(\langle x-z,\Pi_\Omega(x,\xi)\rangle \le \langle x-z,\xi\rangle\).
\end{lemma}

Finally, since the subsequent stability analysis is carried out in incremental coordinates around a steady-state operating point, we recall the notion of incremental passivity adopted in this paper \cite{7912328}. Particularly, when no ambiguity
can arise, we simply write passivity for incremental passivity.

\begin{definition}[Incremental passivity]
Consider a dynamical system with state \(x\), input \(u\), and output \(y\). Let \((x^*,u^*,y^*)\) be an equilibrium triplet. The system is said to be incrementally passive with respect to \((x^*,u^*,y^*)\) if there exists a nonnegative storage function \(S(x,x^*)\) such that, along its trajectories, \(\dot S \le (y-y^*)^\top (u-u^*)\). It is said to be output strictly incrementally passive if there exists \(\alpha>0\) such that \(\dot S \le (y-y^*)^\top (u-u^*) - \alpha \|y-y^*\|^2\).
\end{definition}

\section{Problem Formulation}
\label{sec:problem}

\subsection{Network Plant Dynamics}

Consider a connected transmission network whose topology is described by the graph
\(\mathcal{G}=(\mathcal{N},\mathcal{E})\), where \(n:=|\mathcal{N}|\) is the
number of buses and \(m:=|\mathcal{E}|\) is the number of transmission lines.
The network is partitioned into \(N_a\) control areas indexed by
\(\mathcal{A}:=\{1,\dots,N_a\}\). A line \(\ell\in\mathcal{E}\) is called a tie-line if its two incident buses
belong to different control areas. After fixing an arbitrary orientation for
all lines, let \(C\in\mathbb{R}^{n\times m}\) denote the associated incidence
matrix, and let \(B\) be the diagonal matrix of branch
susceptances. The corresponding weighted Laplacian is then given by
\(L:=CBC^\top\). Moreover, for each ordered area pair \((a,b)\) with
\(a,b\in\mathcal{A}\) and \(a\neq b\), let \(\mathcal{E}_{ab}\subseteq\mathcal{E}\)
denote the set of tie-lines directed from area \(a\) to area \(b\) under the
prescribed orientation. Let
\(\mathcal I_{\mathrm{tie}}\) denote a prescribed collection of ordered area pairs
\((a,b)\) with
\(\mathcal E_{ab}\neq\emptyset\), and let \(n_t:=|\mathcal I_{\mathrm{tie}}|\).

We model the physical power network by the classical linearized swing dynamic (e.g., \cite{andersonPowerSystemControl2002,malladaOptimalLoadSideControl2017,zhaoDesignStabilityLoadSide2014}). Let~\(\theta,\omega,d\in\mathbb{R}^n\) denote, respectively, the vectors of
bus phase angles, frequency deviations, and net demands, where \(d\) is defined as the local load minus the local renewable generation. Let \(n_g\) denote the number of controllable generators, let \(p\in\mathbb{R}^{n_g}\) be the corresponding generator-level control input, and let \(G\in\mathbb{R}^{n\times n_g}\) denote the generator-to-bus selection matrix so that \(G p\in\mathbb{R}^n\) is the nodal controllable injection vector. Such generator-to-bus incidence matrices are standard in direct-current optimal power flow (DC-OPF) and dispatch formulations \cite{FrankRebennack2016,zimmermanMATPOWERSteadyStateOperations2011}. The network dynamics are
\begin{subequations}\label{eq:swing}
\begin{align}
\dot{\theta} &= \omega, \label{eq:swing-theta}\\
M\dot{\omega} &= -D\omega - L\theta + G p - d, \label{eq:swing-omega}
\end{align}
\end{subequations}
where \(M\succ0\) and \(D\succ0\) are the diagonal inertia and damping
matrices, respectively.

At a synchronous steady state with restored nominal frequency, one has
\(\omega^*=0\) and \(\dot\omega=0\). Hence, the steady-state variables
satisfy \(L\theta^* = G p^* - d\).
Thus, at equilibrium, the swing dynamics reduce to a nodal supply-demand balance relation, which motivates the steady-state dispatch problem to be introduced.

\begin{remark}
The model \eqref{eq:swing} is given for a reduced network whose retained
buses are equipped with frequency dynamics.
If the network contains pure load buses, then the physical plant model is more naturally described by a differential-algebraic system. In this case, the
pure load buses may be eliminated via Kron reduction
\cite{dorfler2013kron}, and the
reduced network again admits a swing-dynamics representation of 
\eqref{eq:swing}, with suitably modified network parameters. Accordingly, when
pure load buses are present, the communication graph considered in this paper
should be interpreted as the graph of the Kron-reduced network.
\end{remark}

\begin{remark}
The model \eqref{eq:swing} is a standard approximation in
frequency regulation studies, because it captures the dominant transmission-grid frequency
dynamics with relatively low complexity. The proposed framework is, however,
not restricted to this model. Indeed, the subsequent analysis relies mainly
on incremental passivity and passivity-preserving interconnections. Therefore, the same construction can be extended to transmission-grid models
augmented with turbine-governor dynamics, provided that the corresponding
input-output subsystem is incrementally passive
\cite{yang2020distributed,10319778}. This extension is immediate for
first-order turbine-governor dynamics due to its inherent incremental passivity. For second-order dynamics,
incremental passivity can likewise be established using the arguments in
\cite{7912328}. Hence, the proposed framework applies to a broader class of
transmission-grid models.
\end{remark}

\begin{remark}
The DC power-flow model has been extensively used in frequency-control studies, as it
provides a tractable approximation of active-power transfer around nominal
operating conditions \cite{malladaOptimalLoadSideControl2017}. However, the proposed
framework can, in principle, be extended to more detailed alternating-current (AC) power-flow models with sine-type terms. Such an extension would require a more elaborate Lyapunov construction to handle the resulting nonlinear network coupling by replacing the quadratic DC network energy, to be introduced in
Section~\ref{Sec5}, with  an
incremental energy function associated with the nonlinear AC
potential, so that the sine-type network coupling terms can be absorbed in the
Lyapunov derivative \cite{yang2020distributed,10319778}. To keep the analysis
focused, this paper adopts the DC power-flow model for simplicity.
\end{remark}

\subsection{Steady-State Operational Objectives}

Guided by the steady-state relation of \eqref{eq:swing}, we formulate an
optimization problem whose solutions represent the desired synchronous
operating points of the physical network. Introduce two optimizer-side decision
variables: the controllable generator-dispatch vector
\(u\in\mathbb{R}^{n_g}\) and the associated virtual phase-angle vector
\(\phi\in\mathbf{1}_n^\perp\). Here, \(u\) denotes the target steady-state power
dispatch to be realized by the controller, while \(G u\) is the corresponding nodal controllable injection. The variable \(\phi\)
parameterizes the equilibrium phase-angle profile and is used to encode nodal
power balance, scheduled inter-area exchanges, and line-flow constraints.

To express the scheduled inter-area exchange compactly, define
\(T\in\mathbb{R}^{n_t\times m}\) as the tie-line aggregation matrix. Its row
indexed by \((a,b)\in\mathcal I_{\mathrm{tie}}\) is specified by
\(T_{(a,b),\ell}=1\) if \(\ell\in\mathcal E_{ab}\), and
\(T_{(a,b),\ell}=0\) otherwise.
Then, \(TBC^\top\phi\in\mathbb{R}^{n_t}\) collects the total directed
tie-line flows associated with the ordered area pairs in \(\mathcal I_{\mathrm{tie}}\). Let
\(P^{\mathrm{sch}}\in\mathbb{R}^{n_t}\) denote the corresponding vector of
scheduled inter-area exchanges. The desired steady-state operating point is characterized by the constrained economic dispatch problem:
\begin{subequations}\label{eq:opt}
\begin{align}
\min_{u\in\R^{n_g},\ \phi\in\1_n^\perp}\quad
&J(u)\notag\\
\text{s.t.}\quad
&r(u,\phi):=G u-d-L\phi=0, \label{eq:pb}\\
&g(\phi):=TBC^\top\phi-P^{\mathrm{sch}}=0, \label{eq:tl}\\
&h^+(\phi):=BC^\top\phi-\bar{f}\leq 0, \label{eq:lu}\\
&h^-(\phi):=\underline{f}-BC^\top\phi\leq 0, \label{eq:ld}\\
&u\in[\underline{u},\bar{u}], \label{eq:bc}
\end{align}
\end{subequations}
where \(r(u,\phi)=0\) enforces the nodal power balance at equilibrium,
\(g(\phi)=0\) enforces the prescribed inter-area tie-line exchanges, and
\(h^\pm(\phi)\le 0\) impose transmission-line
flow limits. Here, \(\underline f,\bar f\in\mathbb R^m\) are the lower and upper line-flow
bounds, respectively. The box constraint
\(\Omega:=\{u\in\mathbb{R}^{n_g}\mid \underline u\le u\le \bar u\}\)
specifies the admissible range of controllable power injections.

The formulation \eqref{eq:opt} is constructed so that its feasible solutions
correspond to physically realizable synchronous operating points of
\eqref{eq:swing}. Specifically, at steady state, the optimizer-side dispatch
variable \(u\) is identified with the actual plant input \(p\), while the
virtual angle variable \(\phi\) parameterizes the target physical phase-angle
profile \(\theta\), up to the chosen reference. Under these
identifications, the constraints in \eqref{eq:opt} coincide with the
steady-state equations and operational requirements of the physical network.
In particular, since \(L\mathbf{1}_n=0\), every
\(\phi\in\mathbf{1}_n^\perp\) satisfies
\(\mathbf{1}_n^\top L\phi=0\), and the nodal balance constraint
\(r(u,\phi)=0\) therefore implies
\(\mathbf{1}_n^\top(G u-d)=0\). Thus, the total power-balance condition is
already enforced by the nodal balance equation, and no additional
network-wide balance constraint is needed. Consequently,
every feasible solution of \eqref{eq:opt} yields a synchronous steady-state
operating point compatible with the plant dynamics.

\begin{remark} The objective \(J(u)\) denotes the aggregate generation redispatch cost associated with the target~controllable injections.
In secondary frequency control,  \(J(u)\) is commonly modeled as a strongly convex function, with quadratic costs being a standard choice \cite{malladaOptimalLoadSideControl2017,7912328,yang2020distributed}, which has both engineering and analytical motivations. From an
engineering viewpoint, strong convexity reflects the increasing marginal cost
of regulation effort. Whereas, from an analytical viewpoint, it improves the
well-posedness of the dispatch problem and facilitates the closed-loop
optimality and stability analysis. 
\end{remark}

To characterize the steady-state optimality conditions that the controller
should enforce, we impose the following standard convexity and feasibility
assumption.

\begin{assumption}
\label{ass:convex-feasible}
The objective $J:\R^{n_g}\to\R$ is twice continuously differentiable and
\(m_J\)-strongly convex on \(\Omega\), i.e., for any $u,v\in\Omega$
\begin{equation}
\bigl(\nabla J(u)-\nabla J(v)\bigr)^\top(u-v)\ge m_J\norm{u-v}^2,
\label{eq:strong-convexity}
\end{equation}
for some \(m_J>0\). Moreover, problem \eqref{eq:opt} is strictly feasible,
that is, there exists a point satisfying constraints \eqref{eq:pb}-\eqref{eq:bc} strictly. And at the KKT point, all inequality constraints \eqref{eq:lu}-\eqref{eq:bc}
satisfy strict complementarity.
\end{assumption}

\begin{remark}
The feasibility and complementarity conditions in
Assumption~\ref{ass:convex-feasible} are standard in power-system
optimization and control. Strict feasibility is a natural modeling
requirement; when no strict margin exists, selected limits are often relaxed
or modeled as soft constraints to recover an interior feasible operating
point \cite{malladaOptimalLoadSideControl2017}. Strict complementarity is a standard local nondegeneracy condition \cite{allibhoy2023control}. It
excludes weakly active inequality constraints, namely constraints that are
binding at the KKT point but have zero marginal value. This condition is
vacuous when all inequalities are inactive and can be verified a posteriori
at the computed KKT point.
\end{remark}
Under Assumption~\ref{ass:convex-feasible}, the dispatch problem is a convex
optimization problem for which the KKT conditions are necessary and sufficient.
The next lemma gives the corresponding normal-cone form. This form will be
used later to connect the equilibria of the projected primal--dual controller
with the optimal solutions of \eqref{eq:opt}.
\begin{lemma}[KKT characterization]
\label{prop:kkt} 
A pair
\((u^\star,\phi^\star)\) is optimal for \eqref{eq:opt} if and only if there
exist multipliers
\(\lambda^\star\in\R^n\),~\(\pi^\star\in\R^{n_t}\), and
\(\rho^{\pm\star}\in\R_{\ge 0}^m\) such that $(u^\star,\phi^\star,\lambda^\star,$ $\pi^\star,\rho^{\pm\star})$
satisfies
\begin{subequations}
\label{eq:kkt}
\begin{align}
&r(u^\star,\phi^\star)=0, \ 
g(\phi^\star)=0,\  h^\pm(\phi^\star)\leq 0, \ 
u^\star\in\Omega,\label{eq:pf}\\
&(\rho^{+\star})^\top h^+(\phi^\star)=0,\ (\rho^{-\star})^\top h^-(\phi^\star)=0\label{eq:cs}\\
&0\in \nabla J(u^\star)+G^\top\lambda^\star+\mathcal{C}_{\Omega}(u^\star),\label{eq:nc}\\
&0=-L\lambda^\star+CBT^\top\pi^\star+CB(\rho^{+\star}-\rho^{-\star}).\label{eq:mc}
\end{align}
\end{subequations}
\end{lemma}
\begin{proof}
Writing the set constraint \(u^\star\in\Omega\) through the normal cone \(\mathcal C_\Omega(u^\star)\) \cite{yang2026distributed,BauschkeCombettes2017}, the Lagrangian stationarity conditions with respect to \(u\) and \(\phi\) yield \eqref{eq:nc} and \eqref{eq:mc}, while \eqref{eq:pf}--\eqref{eq:cs} collect the remaining KKT conditions. Therefore, the stated conditions are equivalent to optimality. \qed
\end{proof}

The normal-cone KKT characterization specifies the desired steady-state set.
For the subsequent incremental stability analysis, it is also important to
specify that the primal optimizer is unique, whereas the multipliers may
remain nonunique. Indeed, \eqref{eq:strong-convexity} ensures the uniqueness of the optimal dispatch $u^\star$. Moreover, if
\((u^\star,\phi_1^\star)\) and \((u^\star,\phi_2^\star)\) are both optimal,
then the nodal balance equations give
\(r(u^\star,\phi_1^\star)=r(u^\star,\phi_2^\star)=0\), and hence
\(L(\phi_1^\star-\phi_2^\star)=0\). Since
\(\phi_1^\star,\phi_2^\star\in\1_n^\perp\) and the network is connected, \(L\)
is positive definite on \(\1_n^\perp\), which implies
\(\phi_1^\star=\phi_2^\star\). Therefore, the primal optimal pair is unique,
although the associated KKT multipliers need not be unique, which is
summarized in the following lemma:

\begin{lemma}[Primal uniqueness]
\label{lem:primal-unique}
Under Assumption \ref{ass:convex-feasible}, the optimal pair
\((u^\star,\phi^\star)\) of \eqref{eq:opt} is unique.
\end{lemma}

\subsection{Cyber-Physical Realization and Interface Delays}

The constrained problem \eqref{eq:opt} specifies the desired steady-state
operating point of the closed-loop system.
The control objective is then to design a dynamic feedback mechanism whose equilibria satisfy the optimality
conditions in Lemma~\ref{prop:kkt}. To this end, the physical network is
interconnected with a cyber optimization layer, which updates
the control input in real time based on plant-side measurements. Therefore, \eqref{eq:opt} characterizes the equilibrium to be attained,
whereas the closed-loop dynamics introduced below describe how the cyber and physical layers interact to drive the system toward that equilibrium.

Under ideal delay-free implementation, the cyber layer receives current plant
measurements instantaneously and returns the updated control action
without latency. In practice, nevertheless, these signals are transmitted through
the cyber-physical interface between the physical plant and the
control center.  We thus distinguish two communication delays:
the uplink delay from plant-side measurements to the cyber layer, and the
downlink delay from the cyber layer back to the actuators.

\begin{remark}
The delays considered in this paper arise only from the bidirectional
cyber-physical interface. The internal computation of the cyber optimizer is
assumed to be delay-free. Hence, the delay mechanism studied here is not caused
by the optimizer-side internal dynamics, but by the measurement and actuation
channels connecting the physical plant and the centralized cyber layer.
\end{remark}

We now describe the actual communication signals transmitted across the
cyber--physical interface. On the plant side, the optimizer command \(u\) is
implemented as the generator-level plant input \(p\in\R^{n_g}\). The
measurement returned to the optimizer is denoted by \(y\in\R^{n_g}\), and in
the ideal delay-free limit it coincides with \(G^\top\omega\). Thus, under
ideal delay-free interconnection, one has \(p=u\) and
\(y=G^\top\omega\).

If communication delays are imposed directly on the raw port variables, then
passivity is, in general, no longer preserved \cite{10319778}. To obtain a
delay-compatible passive interconnection, we represent the actual
communication interface in scattering form. Let \(\sigma_p^+\), \(\sigma_o^+\)
be the raw waves launched into the communication channels, whereas let
\(\sigma_p^-\), \(\sigma_o^-\) be the raw waves received from the channels at
the corresponding sides. Here, the subscript \(p\) refers to the plant side of
the interface and the subscript \(o\) refers to the optimizer side. Given a
channel impedance \(\eta>0\), the raw wave variables are defined by
\begin{equation}
\sigma_p^\pm:=\frac{1}{\sqrt{2\eta}}\bigl(p\mp\eta G^\top\omega\bigr),\quad
\sigma_o^\pm:=\frac{1}{\sqrt{2\eta}}\bigl(u\pm\eta y\bigr).
\end{equation}

\begin{figure}[t]
\centering
\includegraphics[width=\linewidth]{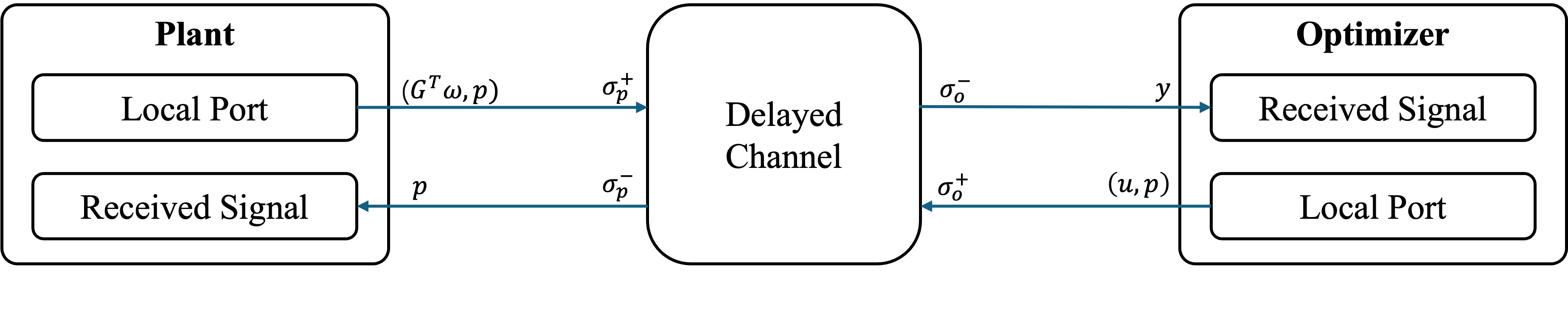}
\caption{Signal interconnections between the plant, optimizer, and scattering-based communication interface.}
\label{fig:plant_optimizer_interface}
\end{figure}

Let \(d_u\) denote the bounded downlink control-to-plant delay and
\(d_\omega\) the uplink plant-to-optimizer delay. The raw delayed
communication channels are represented as 
\begin{equation}
\sigma_p^-(t)=\sigma_o^+(t-d_u),\quad
\sigma_o^-(t)=\sigma_p^+(t-d_\omega).
\end{equation}
Hence, the communication
interface transmits raw waves and does not require knowledge of any
equilibrium quantity. The port variables on the plant and optimizer sides
and the corresponding wave transmissions across the communication interface are given in
Fig.~\ref{fig:plant_optimizer_interface}  . The following optimizer dynamics and passivity analysis are built on this
implementation.
\begin{remark}
The signal \(y\) is not an additional physical measurement sent by the
plant. The actual uplink signal is the wave variable \(\sigma_p^+\), which is
formed at the plant side from \(p\) and \(G^\top\omega\), transmitted through
the measurement channel, and received at the optimizer side as
\(\sigma_o^-(t)=\sigma_p^+(t-d_\omega)\). After receiving \(\sigma_o^-\), the
optimizer combines it with its locally available command \(u\) and reconstructs $y$ using
\(y=(u-\sqrt{2\eta}\sigma_o^-)/\eta\). Similarly, the plant-side input \(p\)
is decoded from the received downlink wave \(\sigma_p^-\) and the local
frequency signal \(G^\top\omega\) using
\(p=\sqrt{2\eta}\sigma_p^- - \eta G^\top\omega\). Thus, \(y\) and \(p\) are not separately transmitted raw signals. Instead,
they are the optimizer-side and plant-side port variables decoded from the
received uplink and downlink waves, respectively. 
\end{remark}

\section{Projected Primal-Dual Dynamics}
\label{sec:controller}

\subsection{Controller Design and Equilibrium Optimality}

The dispatch problem \eqref{eq:opt} specifies the desired steady-state
operating point. We now construct a dynamic optimizer whose equilibria recover
the optimality conditions in Lemma~\ref{prop:kkt}, and which can later be
interconnected with the physical plant through the cyber-physical interface
introduced in Section~\ref{sec:problem}.

Consider the augmented Lagrangian \(\mathcal{L}_a:=\mathcal{L}_a(u,\phi,\lambda,\pi,\) \(\rho^\pm)\) defined by
\begin{equation}
\begin{aligned}
\mathcal{L}_a
:={}&J(u)
+\lambda^\top r(u,\phi)
+\frac{\kappa}{2}\norm{r(u,\phi)}^2+\pi^\top g(\phi) \\
&+(\rho^+)^\top h^+(\phi)
+(\rho^-)^\top h^-(\phi),
\end{aligned}\label{eq:lagrangian}
\end{equation}
where \(\kappa>0\) is an augmentation gain.
The augmentation is applied only to the nodal balance constraint \eqref{eq:pb}. This choice makes the balance
residual appear explicitly in the subsequent passivity argument and improves
the dissipativity properties of the optimizer dynamics. The actuator constraint \(u\in\Omega\), on the other
hand, is enforced directly by
projection, since it represents a
hard feasibility requirement and admits an efficient closed-form
projection in the constrained case \cite{yang2026distributed}.

We model the cyber optimizer as an open dynamical system driven by the
measurement input \(y\) supplied through the cyber-physical interface. Then,
with \eqref{eq:lagrangian}, the projected primal-dual controller is
designed as follows
\begin{subequations}\label{eq:optimizer}
\begin{align}
\tau_u\dot{u}
=\ &
\Pi_{\Omega}
\bigl(
u,
-\nabla J(u)-G^\top\lambda-\kappa G^\top r(u,\phi)-y
\bigr),
\label{eq:optimizer-u} \\
\tau_\phi\dot{\phi}
=\ &
L\lambda
-CB(T^\top\pi+\rho^+-\rho^-)
+\kappa L r(u,\phi), \label{eq:optimizer-phi}\\
\tau_\lambda\dot{\lambda}
=\ &
r(u,\phi),
\label{eq:optimizer-lambda} \\
\tau_\pi\dot{\pi}
=\ &
g(\phi),
\label{eq:optimizer-pi} \\
\tau_+\dot{\rho}^+
=\ &
\Pi_{\R_{\geq 0}^m}\bigl(\rho^+,h^+(\phi)\bigr),
\label{eq:optimizer-rhop} \\
\tau_-\dot{\rho}^-
=\ &
\Pi_{\R_{\geq 0}^m}\bigl(\rho^-,h^-(\phi)\bigr),
\label{eq:optimizer-rhom}
\end{align}
\end{subequations}
where
\(\tau_u,\tau_\phi,\tau_\lambda,\tau_\pi,\tau_+,\tau_-\succ 0\) are diagonal
gain matrices. In \eqref{eq:optimizer}, the \(u\)-subsystem updates
the dispatch variable subject to the capacity constraint, the
\(\phi\)-subsystem updates the virtual phase-angle state to enforce network
consistency, the \(\lambda\)- and \(\pi\)-subsystems integrate the
equality-constraint residuals, and the \(\rho^\pm\)-subsystems enforce the
line-flow inequalities through projected multiplier dynamics. Hence, under feasible initialization, i.e.,
\(u(0)\in\Omega\) and \(\rho^\pm(0)\in\R_{\ge0}^m\), the dynamics preserves
\(u\in\Omega\) and \(\rho^\pm\in\R_{\ge0}^m\). When \(y=0\),
\eqref{eq:optimizer} reduces to a standard augmented projected primal-dual
flow associated with \eqref{eq:opt}; after interconnection with the physical
plant, \(y\) supplies the frequency-measurement signal injected through the
communication interface.
\begin{remark}
In practical power systems, the control center typically has access to
standard operational information, including frequency measurements, tie-line
power, generator outputs, operating limits, and network model data
\cite{NERC2024EMSRisksMitigations,NERC2021BalancingFrequencyControl,PJM2025Manual01}.
Traditional AGC uses only a low-dimensional part of this information, mainly
through the Area Control Error, and is therefore primarily aimed at power
balancing and frequency restoration
\cite{NERC2021BalancingFrequencyControl}. In contrast, the proposed
primal-dual controller uses the available telemetry and network model more
systematically within an optimization-and-control loop, thereby coordinating
frequency regulation, tie-line schedule restoration, transmission congestion
management, and generation capacity constraints in a unified manner.
\end{remark}
\begin{remark}
The projection step in \eqref{eq:optimizer-u} is not limited to componentwise box constraints. Its role is to enforce convex actuator-feasibility constraints
directly at the primal-update level. The same projection-based design can also accommodate coupled output constraints across multiple units, aggregate regulation-capacity limits at the plant or area level \cite{Wang2019DistributedFrequencyControlPartI}, dynamic output-feasibility requirements arising from regulation scheduling \cite{Zhang2019AGCDynamicsConstrainedED}, and participation-factor or reserve-coupling constraints among multiple resources \cite{Zhang2017DRCCOPFReserves}. The box-constrained formulation is therefore adopted for notational simplicity and implementation convenience, rather than as a structural limitation of the method.
\end{remark}

\begin{remark}
When the feasible set has a complicated geometry, computing the projection
can be computationally demanding, and projection-free methods, such as
Frank-Wolfe algorithms, are often used to reduce the online computational
burden by replacing projection with a linear minimization oracle
\cite{Jiang2025DistributedStochasticProjectionFree,wai2017decentralizedFW}. However, such a
situation does not arise in \eqref{eq:optimizer}, where the only projections are onto the box set \(\Omega\) and the nonnegative
orthant \(\R_{\ge0}^m\), both of which reduce to componentwise truncations. Moreover, for
problems involving more complex feasible regions, the present framework could
in principle be extended by replacing the projected primal update with a
Frank-Wolfe-type update. The corresponding analysis would then need to
quantify whether the energy perturbation introduced by the projection-free
approximation can be absorbed in the passivity or Lyapunov argument. This
extension is left for future work.
\end{remark}
The following proposition confirms that the projected dynamics are consistent with
the steady-state dispatch objective. To be specific, when the measurement input
vanishes, every equilibrium of \eqref{eq:optimizer}
satisfies the KKT conditions of \eqref{eq:opt}.  It is worth pointing out that since the equilibria of \eqref{eq:optimizer} and the KKT points of \eqref{eq:opt} are equivalent, we use the superscript \(\star\) for both in the sequel.

\begin{proposition}[Optimality of equilibria]
\label{prop:equilibrium-optimality}
If
\((u^\star,\phi^\star,\lambda^\star,\pi^\star,\rho^{\pm \star})\) is an
equilibrium of \eqref{eq:optimizer}, then \((u^\star,\phi^\star)\) is an
optimal solution of problem \eqref{eq:opt}.
\end{proposition}
\begin{proof}
At equilibrium, \eqref{eq:optimizer-lambda} and \eqref{eq:optimizer-pi} give \eqref{eq:pb} and \eqref{eq:tl}. Applying Lemma~\ref{lem:proj-equilibrium} to \eqref{eq:optimizer-u} and \eqref{eq:optimizer-rhop}--\eqref{eq:optimizer-rhom} gives \eqref{eq:pf}, \eqref{eq:cs}, and \eqref{eq:nc}, and \eqref{eq:optimizer-phi} gives \eqref{eq:mc}. Therefore, by Lemma~\ref{prop:kkt}, \((u^\star,\phi^\star)\) is optimal for \eqref{eq:opt}.
\end{proof}

\subsection{Passivity-Based Closed-Loop Analysis}\label{Sec5}

Having established the steady-state consistency of the cyber optimizer, we now
analyze its dynamic interconnection with the physical network through the
delayed cyber-physical interface. By Lemma~\ref{lem:primal-unique}, the primal
optimizer \((u^\star,\phi^\star)\) is unique.
Fix a representative KKT point
\(
z^\star:=(u^\star,\phi^\star,\lambda^\star,\pi^\star,\rho^{\pm\star}).
\)
The associated physical equilibrium is chosen as
\(\theta^\star=\phi^\star\), \(\omega^\star=0\), and \(p^\star=u^\star\)
after fixing the angle reference so that the angle variables lie in
\(\1_n^\perp\).  Define the
incremental variables by
\(\tilde\theta:=\theta-\theta^\star\),
\(\tilde\omega:=\omega-\omega^\star\),
\(\tilde u:=u-u^\star\), and define
\(\tilde\phi,\tilde\lambda,\tilde\pi,\tilde\rho^\pm, \tilde p\) analogously. 
Since \(\omega^\star=0\), one has \(G^\top\omega=G^\top\tilde\omega\) in
these coordinates, and the corresponding raw equilibrium waves satisfy $\sigma_p^{+\star}=\sigma_p^{-\star}
=\sigma_o^{+\star}=\sigma_o^{-\star}
=\frac{u^\star}{\sqrt{2\eta}}.$
Define the incremental wave variables by
\begin{equation}
\begin{aligned}
\tilde{\sigma}_p^\pm&:=\sigma_p^\pm-\frac{u^\star}{\sqrt{2\eta}}
=\frac{1}{\sqrt{2\eta}}\bigl(\tilde p\mp\eta G^\top\tilde\omega\bigr),\\
\tilde{\sigma}_o^\pm&:=\sigma_o^\pm-\frac{u^\star}{\sqrt{2\eta}}
=\frac{1}{\sqrt{2\eta}}\bigl(\tilde u\pm\eta y\bigr).
\end{aligned}
\label{eq:spm}
\end{equation}
Since the raw delay channels are linear and the equilibrium wave bias is
constant, the incremental waves satisfy
\begin{equation}
\tilde{\sigma}_p^-(t)=\tilde{\sigma}_o^+(t-d_u),\quad
\tilde{\sigma}_o^-(t)=\tilde{\sigma}_p^+(t-d_\omega).
\label{eq:delay-lines}
\end{equation}
The analysis is conducted componentwise: the plant is viewed as a system from
\(\tilde p\) to \(G^\top\tilde\omega\), the optimizer is viewed as a system
from \(y\) to \(-\tilde u\), and the scattering channel is viewed as a
lossless passive transmission element. The componentwise passivity properties
are then combined to obtain a closed-loop dissipation inequality.
In incremental coordinates, the plant dynamics are
\begin{equation}\label{eq:plant-dev}
\begin{aligned}
\dot{\tilde{\theta}} &= \tilde{\omega},\\
M\dot{\tilde{\omega}} &= -D\tilde{\omega}-L\tilde{\theta}+G\tilde p.
\end{aligned}
\end{equation}
The following lemma collects the passivity properties of the plant, optimizer,
and delayed scattering channel at their respective interface ports.

\begin{lemma}[Componentwise passivity]
\label{lem:passivity-blocks}
The following passivity properties hold.

\begin{enumerate}
\item[(i)]
The plant subsystem \eqref{eq:plant-dev} is passive from \(\tilde p\) to \(G^\top\tilde\omega\), with strict damping in \(\tilde{\omega}\).

\item[(ii)] 
The optimizer subsystem \eqref{eq:optimizer} is passive from \(y\) to \(-\tilde u\). With
the quadratic incremental storage used below, its dissipation
contains the gradient-monotonicity term and the augmented power-balance residual.

\item[(iii)] 
The delayed scattering channel \eqref{eq:spm}-\eqref{eq:delay-lines} is lossless passive, and its storage
accounts for the energy transported along the delay lines.
\end{enumerate}
\end{lemma}

\begin{proof}
(i) Consider
\(S_p:=\frac{1}{2}\tilde{\omega}^\top M\tilde{\omega}
+\frac{1}{2}\tilde{\theta}^\top L\tilde{\theta}\).
Along the trajectories of \eqref{eq:plant-dev},
\begin{equation}\label{passive1}
\dot S_p=-\tilde{\omega}^\top D\tilde{\omega}
+\tilde{\omega}^\top G\tilde p
=(G^\top\tilde\omega)^\top\tilde p-\tilde{\omega}^\top D\tilde{\omega},
\end{equation}
which shows passivity from \(\tilde p\) to \(G^\top\tilde{\omega}\), with
strict damping in \(\tilde{\omega}\).

(ii) Define the optimizer storage
\[
\begin{aligned}
S_o
:={}&
\frac{1}{2}\tilde{u}^\top\tau_u\tilde{u}
+\frac{1}{2}\tilde{\phi}^\top\tau_\phi\tilde{\phi}
+\frac{1}{2}\tilde{\lambda}^\top\tau_\lambda\tilde{\lambda}+
\frac{1}{2}\tilde{\pi}^\top\tau_\pi\tilde{\pi}\\
&
+\frac{1}{2}\tilde{\rho}^{+\top}\tau_+\tilde{\rho}^+
+\frac{1}{2}\tilde{\rho}^{-\top}\tau_-\tilde{\rho}^-.
\end{aligned}
\]
Let
\(\xi_u:=-\nabla J(u)-G^\top\lambda-\kappa G^\top r(u,\phi)-y\) for notation simplicity. At the reference equilibrium with
\(y=0\), define
\(\xi_u^\star:=-\nabla J(u^\star)-G^\top\lambda^\star\).
Since \((u^\star,\phi^\star,\lambda^\star,\pi^\star,\rho^{\pm\star})\) is an
equilibrium of \eqref{eq:optimizer} with \(y=0\), Lemma~\ref{lem:proj-equilibrium} gives~\(\xi_u^\star\in \mathcal{C}_\Omega(u^\star)\),
and \(\tilde{u}^\top \xi_u^\star\le 0\). By Lemma~\ref{lem:proj-nonexpansive}, applied to \eqref{eq:optimizer-u},
\begin{equation}
\label{eq:u-block-rewrite}
\tilde{u}^\top\tau_u\dot{u}
=
\tilde{u}^\top\Pi_\Omega(u,\xi_u)
\le
\tilde{u}^\top(\xi_u-\xi_u^\star).
\end{equation}
Lemma~\ref{lem:proj-equilibrium} implies
\(h^\pm(\phi^\star)\in \mathcal C_{\R_{\ge0}^m}(\rho^{\pm\star})\), and hence
\(\tilde{\rho}^{\pm\top}h^\pm(\phi^\star)\le 0\). Applying
Lemma~\ref{lem:proj-nonexpansive} to
\eqref{eq:optimizer-rhop}-\eqref{eq:optimizer-rhom} gives
\begin{equation}
\tilde{\rho}^{+\top}\tau_+\dot{\rho}^+
+\tilde{\rho}^{-\top}\tau_-\dot{\rho}^-
\le
(\tilde{\rho}^+-\tilde{\rho}^-)^\top BC^\top\tilde{\phi}.
\end{equation}
Moreover, the \((\phi)\)-, \((\lambda)\)-, \((\pi)\)-blocks in \eqref{eq:optimizer-phi}-\eqref{eq:optimizer-pi} satisfy
\begin{equation}
\begin{aligned}
&\tilde{\phi}^\top\tau_\phi\dot{\phi}
+\tilde{\lambda}^\top\tau_\lambda\dot{\lambda}
+\tilde{\pi}^\top\tau_\pi\dot{\pi}\\
=&
\tilde{u}^\top G^\top\tilde{\lambda}
+\kappa \tilde{\phi}^\top L r(u,\phi)-\tilde{\phi}^\top CB(\tilde{\rho}^+-\tilde{\rho}^-).
\end{aligned}
\end{equation}
Since
\(r(u^\star,\phi^\star)=0\), one has \(r(u,\phi)=G\tilde u-L\tilde\phi\), and therefore
\begin{equation}\label{kappa}
-\kappa \tilde u^\top G^\top r(u,\phi)
+\kappa \tilde\phi^\top L r(u,\phi)
=
-\kappa \norm{r(u,\phi)}^2.
\end{equation}
Combining \eqref{eq:u-block-rewrite}-\eqref{kappa} yields
\begin{equation}
\label{eq:optimizer-passivity-wq}
\dot{S}_o
\leq
-\tilde{u}^\top\bigl(\nabla J(u)-\nabla J(u^\star)\bigr)
-\kappa\norm{r(u,\phi)}^2
-\tilde u^\top y.
\end{equation}
This proves passivity of the optimizer from \(y\) to \(-\tilde u\), with the stated
dissipation terms.

(iii) Define the delay-line storage
\[
S_\tau:=
\frac{1}{2}\int_{t-d_u}^{t}\norm{\tilde{\sigma}_o^+(\sigma)}^2\,d\sigma
+
\frac{1}{2}\int_{t-d_\omega}^{t}\norm{\tilde{\sigma}_p^+(\sigma)}^2\,d\sigma.
\]
Differentiating $S_\tau$ and using \eqref{eq:delay-lines} gives
\(
\dot{S}_\tau
=
\frac{1}{2}\bigl(\norm{\tilde{\sigma}_p^+}^2-\norm{\tilde{\sigma}_p^-}^2\bigr)
+
\frac{1}{2}\bigl(\norm{\tilde{\sigma}_o^+}^2-\norm{\tilde{\sigma}_o^-}^2\bigr).
\)
By the scattering definitions \eqref{eq:spm},
\(\frac{1}{2}\bigl(\norm{\tilde{\sigma}_p^+}^2-\norm{\tilde{\sigma}_p^-}^2\bigr)=-(G^\top\tilde\omega)^\top\tilde p\) and
\(\frac{1}{2}\bigl(\norm{\tilde{\sigma}_o^+}^2-\norm{\tilde{\sigma}_o^-}^2\bigr)=\tilde u^\top y\).
Hence
\begin{equation}\label{passive3}
    \dot{S}_\tau=-(G^\top\tilde\omega)^\top\tilde p+\tilde u^\top y.
\end{equation}
which proves lossless passivity of the delayed scattering channel.
\qed
\end{proof}

Lemma~\ref{lem:passivity-blocks} provides the three ingredients required for
the delayed closed-loop analysis. Then, combining these facts gives the
following dissipation inequality for the full delayed interconnection.

\begin{theorem}[Delayed closed-loop dissipation]
\label{thm:raw-delayed-scattering}
Let the plant \eqref{eq:plant-dev}, the optimizer
\eqref{eq:optimizer}, and the scattering channel
\eqref{eq:spm}-\eqref{eq:delay-lines} be interconnected through the
relations defined above. Then, \(S:=S_p+S_o+S_\tau\) satisfies
\begin{equation}
\begin{aligned}
\dot{S}
\leq{}&
-\tilde{\omega}^\top D\tilde{\omega}
-\tilde{u}^\top\bigl(\nabla J(u)-\nabla J(u^\star)\bigr)
-\kappa\norm{r(u,\phi)}^2\\
\leq{}&
-\lambda_{\min}(D)\norm{\tilde\omega}^2
-m_J\norm{\tilde u}^2
-\kappa\norm{r(u,\phi)}^2\\
\leq{}&
-c_0\bigl(
\norm{\tilde\omega}^2+\norm{\tilde u}^2+\norm{r(u,\phi)}^2
\bigr),
\end{aligned}
\label{eq:raw-delayed-dissipation}
\end{equation}
where \(c_0:=\min\{\lambda_{\min}(D),m_J,\kappa\}>0\). Moreover, there exists \(c_\phi>0\) such that
\begin{equation}
\norm{\tilde\phi}^2
\le
c_\phi\bigl(\norm{\tilde u}^2+\norm{r(u,\phi)}^2\bigr).
\label{eq:phi-bound-primal}
\end{equation}
\end{theorem}

\begin{proof}
The first inequality in \eqref{eq:raw-delayed-dissipation} follows by summing \eqref{passive1},
\eqref{eq:optimizer-passivity-wq}, and \eqref{passive3}. The strong-convexity \eqref{eq:strong-convexity} gives
\(\tilde{u}^\top\bigl(\nabla J(u)-\nabla J(u^\star)\bigr)
\ge
m_J\norm{\tilde u}^2\),
which yields the second inequality. The
third inequality follows from the definition of \(c_0\). By
\(r(u^\star,\) \(\phi^\star)=0\), one has
\(
L\tilde\phi=G\tilde u-r(u,\phi).
\)
Since the angle reference is fixed in \(\1_n^\perp\) and the network
is connected, \(L\) is positive definite on \(\1_n^\perp\). Thus,
\eqref{eq:phi-bound-primal} holds for some \(c_\phi>0\).
\qed
\end{proof}

Theorem~\ref{thm:raw-delayed-scattering} shows that the delayed closed loop retains a strict dissipation mechanism in the physical frequency, the primal dispatch variable, and the power-balance residual. We next show that the proposed controller not only guarantees asymptotic stability, but also admits local exponential stability.

Let \(\mathcal R\) denote the local equilibrium set of the nominal delayed
closed-loop system associated with \(z^\star\). This set allows the possible
nonuniqueness of KKT multipliers, while the plant variables and the primal
optimizer variables are fixed at their equilibrium values.

\begin{proposition}[Local exponential stability]
\label{prop:nominal-set-exp}
The nominal delayed
closed-loop system is locally exponentially stable with respect to
\(\mathcal R\).
\end{proposition}

\begin{proof}
By Assumption~\ref{ass:convex-feasible}, strict complementarity implies that,
after restricting the analysis to a sufficiently small neighborhood of
\(z^\star\), the active box constraints and active line-flow constraints are
unchanged. Hence the projected dynamics reduce locally to a smooth system on
the free generator coordinates and the active multiplier coordinates,
interconnected with the same finite constant-delay scattering channels.

Theorem \ref{thm:raw-delayed-scattering} characterizes
the zero-dissipation invariant set of the nominal delayed interconnection as
the KKT equilibrium set. The same characterization applies to the local
linearization, because the linearized system preserves the same passive
plant-optimizer-channel interconnection and the same fixed active set.
Therefore, the linearized dynamics converges to the local linearized
equilibrium set \cite{SepulchreStan2005}. By Lemma~\ref{lem:primal-unique}, the plant variables and the
primal optimizer variables are fixed on this set; the only possible freedom is
in the dual multipliers, and admissible dual variations must leave the same
primal KKT point unchanged. Thus, the only neutral directions are tangent to
\(\mathcal R\).

On the quotient space transverse to \(\mathcal R\), the local linearization is
an autonomous linear system with constant coefficients and finite constant
delays. For such a linear system, asymptotic convergence is equivalent to
exponential convergence \cite{Khalil2002NonlinearSystems}. Hence the transverse distance to \(\mathcal R\)
decays exponentially for the linearized dynamics. By the standard
linearization principle, after possibly shrinking the neighborhood, the
nonlinear nominal delayed closed-loop system is locally exponentially stable
with respect to \(\mathcal R\).
\qed
\end{proof}

For the subsequent RBC implementation and its convergence analysis, we next
pass to a local sampled-data description of the nominal delayed closed loop.
The same equilibrium set \(\mathcal R\) is used throughout this local
discrete-time analysis, and all distances and neighborhoods below are taken
with respect to this set. By Proposition~\ref{prop:nominal-set-exp} and the
converse Lyapunov theorem \cite{Khalil2002NonlinearSystems}, after restricting the analysis to a sufficiently
small neighborhood of \(\mathcal R\), there exist a local Lyapunov function
\(V\) and constants \(a_1,a_2,a_3>0\) such that, for every nominal sampled
trajectory initialized in this neighborhood,
\begin{equation}
a_1\gamma_k^2
\le
V_k
\le
a_2\gamma_k^2,
\label{eq:V-flow-bounds}
\end{equation}
and
\begin{equation}
V_{k+1}-V_k
\le
-a_3\varepsilon\gamma_k^2
\label{eq:V-flow-decay}
\end{equation}
for all sufficiently small \(\varepsilon\), where \(\gamma_k\) denotes the
distance from the current sampled closed-loop state to \(\mathcal R\).
In what follows, \(\varepsilon\) is used as the discrete step size.

\section{RBC Implementation and Analysis}
\label{sec:rbc}

\subsection{RBC implementation}

In a digital implementation of \eqref{eq:optimizer}, updating all control
variables at every sampling instant may be computationally burdensome in
large-scale networks. We therefore use an RBC
implementation of the cyber optimizer. The physical plant and the scattering
interface are unchanged; only the optimizer update is replaced by a sampled
block update. Thus, the RBC scheme is a computationally cheaper realization of
the same closed-loop controller, rather than a different steady-state control
law.

Define the optimizer state
\(z:=\operatorname{col}(u,\phi,\lambda,\pi,\rho^\pm)\). The continuous-time
optimizer can be written compactly as \(\dot z=f(z,y)\), where \(f\) denotes
the vector field induced by \eqref{eq:optimizer}. Let \(t_k=k\varepsilon\) and
\(y_k:=y(t_k)\). Partition \(z\) into \(n_b\) coordinate blocks as
\(z=\operatorname{col}(z^{(1)},\dots,z^{(n_b)})\), and let \(E_\beta\) denote
the coordinate projector associated with block \(\beta\), with
\(\sum_{\beta=1}^{n_b}E_\beta=I\). At time \(t_k\), a block index \(\beta_k\)
is drawn according to \(\mathbb P(\beta_k=\beta)=p_\beta\), where
\(p_\beta>0\) and \(\sum_{\beta=1}^{n_b}p_\beta=1\). The implemented RBC
update is
\begin{equation}
\label{eq:rbc-candidate-delayed}
z_{k+1}
=
\mathcal P_\mathcal K
\bigl(z_k+\varepsilon p_{\beta_k}^{-1}E_{\beta_k}f(z_k,y_k)\bigr),
\end{equation}
where \(\mathcal K:=\Omega\times\1_n^\perp\times\R^n\times\R^{n_t}\times\R_{\ge0}^m\times\R_{\ge0}^m\) is the Cartesian product of the feasibility sets of \eqref{eq:optimizer}. Hence, only one optimizer block is updated at each sampling
instant, while feasibility is preserved by projection. In a sufficiently small neighborhood of \(\mathcal R\),
\(\mathcal P_{\mathcal K}\) acts as the identity on the reduced local update
and can therefore be omitted in the local convergence analysis.

Let \(\mathcal F_k\) denote the history generated by the previous block
selections and the corresponding sampled closed-loop states up to time \(t_k\),
and let \(\E_k[\cdot]:=\E[\cdot\mid\mathcal F_k]\). Since
\(\sum_{\beta=1}^{n_b}E_\beta=I\) and
\(\mathbb P(\beta_k=\beta)=p_\beta\), the scaling in
\eqref{eq:rbc-candidate-delayed} gives
\(\E_k[p_{\beta_k}^{-1}E_{\beta_k}]
=\sum_{\beta=1}^{n_b}p_\beta p_\beta^{-1}E_\beta=I\). Multiplying this identity
by \(f(z_k,y_k)\) gives the standard unbiasedness property summarized below.

\begin{lemma}[Conditional unbiasedness]
\label{lem:rbc-unbiased}
The RBC block direction satisfies
\begin{equation}
\begin{aligned}
\E_k\!\left[p_{\beta_k}^{-1}E_{\beta_k}\right]
&=
I,
\\
\E_k\!\left[p_{\beta_k}^{-1}E_{\beta_k}f(z_k,y_k)\right]
&=
f(z_k,y_k).
\end{aligned}
\label{eq:rbc-unbiased-delayed}
\end{equation}
\end{lemma}

\subsection{Small-perturbation analysis}

We now analyze the RBC sampled trajectory inside the local neighborhood of $\mathcal{R}$. In this neighborhood,
the optimizer vector field is locally Lipschitz and vanishes on
\(\mathcal R\). Consequently, one RBC step differs from the corresponding
nominal sampled step only by a higher-order random implementation error. The
following lemma records this perturbation effect directly at the Lyapunov
level.

\begin{lemma}[One-step RBC estimate]
\label{lem:rbc-perturbation}
Shrinking the local neighborhood if necessary, there exists a constant
\(C_s>0\) such that every RBC sampled trajectory in this neighborhood satisfies
\begin{equation}
\E_k[V_{k+1}]
\le
V_k
-a_3\varepsilon\gamma_k^2
+
C_s\varepsilon^2\gamma_k^2
\label{eq:rbc-one-step-error}
\end{equation}
for all sufficiently small \(\varepsilon\).
\end{lemma}

\begin{proof}
By Proposition~\ref{prop:nominal-set-exp}, the nominal sampled dynamics admit
the local Lyapunov estimate \eqref{eq:V-flow-bounds}--\eqref{eq:V-flow-decay}
in a sufficiently small neighborhood of the equilibrium set. After restricting
the analysis to this neighborhood, the active constraints remain unchanged,
the sampled closed-loop vector field is locally Lipschitz, and it vanishes on
the equilibrium set. Hence, along the sampled trajectory, the one-step
variation of the closed-loop state is \(O(\varepsilon\gamma_k)\).

We apply Taylor's formula to \(V\) along one RBC step. Note that \(V\) is twice
continuously differentiable in the chosen neighborhood and its second
derivative is locally bounded, the RBC update gives
\begin{equation}
    V_{k+1}=V_k+\varepsilon\nabla V^\top
p_{\beta_k}^{-1}E_{\beta_k}f(z_k,y_k)+O(\varepsilon^2\gamma_k^2),
\end{equation}
where
\(\nabla V\) is evaluated at the current sampled state. Taking conditional
expectation and using Lemma~\ref{lem:rbc-unbiased} yields
\(\E_k[V_{k+1}]
\le V_k+\varepsilon\nabla V^\top f(z_k,y_k)+O(\varepsilon^2\gamma_k^2)\).
The term \(O(\varepsilon^2\gamma_k^2)\) follows from
\(f(z_k,y_k)=O(\gamma_k)\) and the boundedness of
\(p_{\beta_k}^{-1}E_{\beta_k}\) under the fixed block probabilities.

The first-order term is the same as the first-order term of the full sampled
nominal update. Applying Taylor's formula to that full sampled update and
using the nominal decrease estimate \eqref{eq:V-flow-decay} gives
\begin{equation}
V_k+\varepsilon\nabla V^\top f(z_k,y_k)\le V_k-a_3\varepsilon\gamma_k^2+O(\varepsilon^2\gamma_k^2).
\end{equation}
Combining the preceding two estimates and increasing \(C_s\) if necessary
gives \eqref{eq:rbc-one-step-error}.
\qed
\end{proof}

\subsection{Mean-square stability of the RBC implementation}

Lemma~\ref{lem:rbc-perturbation} shows that the RBC sampling error is of order
\(O(\varepsilon^2\gamma_k^2)\), while the nominal local decrease is of order
\(O(\varepsilon\gamma_k^2)\). Thus, for sufficiently small step size, the
nominal decrease dominates the RBC sampling error.

\begin{theorem}[Mean-square stability]
\label{thm:rbc-closed-loop}
There exist \(\bar{\varepsilon}>0\), \(C>0\), \(\nu>0\), and a sufficiently
small neighborhood of \(\mathcal R\) such that, for any
\(0<\varepsilon<\bar{\varepsilon}\), each RBC sampled trajectory of \eqref{eq:rbc-candidate-delayed} remains in the local neighborhood where the
active constraints are unchanged and satisfies
\begin{equation}
\E[\gamma_k^2]
\le
C(1-\nu\varepsilon)^k\gamma_0^2,
\label{eq:rbc-mss}
\end{equation}
where, \(\gamma_0\) denotes the initial distance from the RBC sampled
closed-loop state to \(\mathcal R\).
\end{theorem}

\begin{proof}
Choose \(\bar{\varepsilon}>0\) such that
\(C_s\bar{\varepsilon}\le a_3/2\). Then, for any
\(0<\varepsilon<\bar{\varepsilon}\), the perturbation term $C_s\varepsilon^2\gamma_k^2$ in
\eqref{eq:rbc-one-step-error} is dominated by one half of the nominal
Lyapunov decrease. Hence
\begin{equation}
    \E_k[V_{k+1}]
\le
V_k-a_3\varepsilon\gamma_k^2
+C_s\varepsilon^2\gamma_k^2
\le
V_k-\frac{a_3}{2}\varepsilon\gamma_k^2 .
\end{equation}
This is the point where the nominal exponential decrease absorbs the
second-order RBC perturbation.
Using \eqref{eq:V-flow-bounds} gives $\E_k[V_{k+1}]
\le
V_k-\frac{a_3}{2a_2}\varepsilon V_k$.
Thus the conditional expected Lyapunov value contracts by a fixed factor at
each step. Taking total expectation and applying this estimate recursively
gives $\E[V_k]
\le
\left(1-\frac{a_3}{2a_2}\varepsilon\right)^k V_0 .$
With \(\nu:=a_3/(2a_2)\), this becomes $\E[V_k]\le (1-\nu\varepsilon)^kV_0 .$
Finally, the lower and upper bounds in \eqref{eq:V-flow-bounds} give
\(a_1\gamma_k^2\le V_k\) and \(V_0\le a_2\gamma_0^2\). Hence $\E[\gamma_k^2]
\le
\frac{1}{a_1}\E[V_k]
\le
\frac{a_2}{a_1}(1-\nu\varepsilon)^k\gamma_0^2 .$
This proves \eqref{eq:rbc-mss} with \(C=a_2/a_1\).
\qed
\end{proof}
\section{Passivity-Preserving Wave-Domain Interface Filtering}
\label{sec:filtering}

The preceding section developed a sampled-data RBC implementation that reduces
the per-step computational burden of the cyber layer while preserving the
target equilibrium set. The resulting blockwise updates, however, may introduce
fast variations in the optimizer-side wave variables. By
Lemma~\ref{lem:passivity-blocks}(iii), the delayed scattering channel
\eqref{eq:delay-lines} is lossless passive. It preserves the passivity
interconnection under constant delays, but does not attenuate such
high-frequency wave components. This motivates an additional interface
filtering mechanism.

To introduce extra
interface dissipation without changing the steady-state interconnection, we
place first-order filters in the wave domain. Specifically, we replace
\eqref{eq:delay-lines} by
\begin{subequations}\label{eq:filtered-wave-channel}
\begin{align}
a_u(t)&:=\sigma_o^+(t-d_u), \label{eq:filtered-wave-down-delay}\\
\zeta_u\dot \sigma_p^-(t)&=-\sigma_p^-(t)+a_u(t),
\label{eq:filtered-wave-down-filter}\\
a_\omega(t)&:=\sigma_p^+(t-d_\omega), \label{eq:filtered-wave-up-delay}\\
\zeta_\omega\dot \sigma_o^-(t)&=-\sigma_o^-(t)+a_\omega(t),
\label{eq:filtered-wave-up-filter}
\end{align}
\end{subequations}
where \(\zeta_u,\zeta_\omega>0\) are the filter time constants. Both filters
are first-order low-pass filters with unit DC gain. As a result, constant wave
signals are transmitted without steady-state distortion, and the equilibrium
wave interconnection is left unchanged. In the analysis below, we use the
corresponding incremental form of \eqref{eq:filtered-wave-channel}.

\begin{lemma}[Filtered channel passivity]
\label{lem:filtered-wave-channel}
Consider the filtered scattering channel defined by \eqref{eq:spm} and
\eqref{eq:filtered-wave-channel}. In incremental coordinates, let
\(\tilde a_u(t):=\tilde{\sigma}_o^+(t-d_u)\) and
\(\tilde a_\omega(t):=\tilde{\sigma}_p^+(t-d_\omega)\). Define
\[
\begin{aligned}
S_{\mathrm{ch}}
:={}&
\frac{1}{2}\int_{t-d_u}^{t}\norm{\tilde{\sigma}_o^+(\sigma)}^2\,d\sigma
+
\frac{1}{2}\int_{t-d_\omega}^{t}\norm{\tilde{\sigma}_p^+(\sigma)}^2\,d\sigma\\
&+
\frac{\zeta_u}{2}\norm{\tilde{\sigma}_p^-}^2
+
\frac{\zeta_\omega}{2}\norm{\tilde{\sigma}_o^-}^2.
\end{aligned}
\]
Then
\begin{equation}
\dot S_{\mathrm{ch}}
=
-(G^\top\tilde\omega)^\top\tilde p
+\tilde u^\top y
-\frac{1}{2}\norm{\tilde a_u-\tilde{\sigma}_p^-}^2
-\frac{1}{2}\norm{\tilde a_\omega-\tilde{\sigma}_o^-}^2.
\label{eq:filtered-wave-channel-passivity}
\end{equation}
Hence the filtered scattering channel is passive. Compared with the unfiltered
delayed channel, it introduces the additional dissipation terms
\(\frac{1}{2}\norm{\tilde a_u-\tilde{\sigma}_p^-}^2\) and
\(\frac{1}{2}\norm{\tilde a_\omega-\tilde{\sigma}_o^-}^2\).
\end{lemma}

\begin{proof}
Since the equilibrium waves are constant, \eqref{eq:filtered-wave-channel}
gives
\(\zeta_u\dot{\tilde{\sigma}}_p^-=-\tilde{\sigma}_p^-+\tilde a_u\) and
\(\zeta_\omega\dot{\tilde{\sigma}}_o^-=-\tilde{\sigma}_o^-+\tilde a_\omega\).
Differentiating \(S_{\mathrm{ch}}\) along these incremental filter dynamics
gives
\[
\begin{aligned}
\dot S_{\mathrm{ch}}
={}&
\frac{1}{2}\norm{\tilde{\sigma}_o^+}^2
-\frac{1}{2}\norm{\tilde a_u}^2
+
\frac{1}{2}\norm{\tilde{\sigma}_p^+}^2
-\frac{1}{2}\norm{\tilde a_\omega}^2  \\
&+
\tilde{\sigma}_p^{-\top}(\tilde a_u-\tilde{\sigma}_p^-)
+
\tilde{\sigma}_o^{-\top}(\tilde a_\omega-\tilde{\sigma}_o^-).
\end{aligned}
\]
Using
\(b^\top(a-b)=\frac{1}{2}\norm{a}^2-\frac{1}{2}\norm{b}^2-\frac{1}{2}\norm{a-b}^2\)
yields
\[
\begin{aligned}
\dot S_{\mathrm{ch}}
={}&
\frac{1}{2}\bigl(\norm{\tilde{\sigma}_o^+}^2-\norm{\tilde{\sigma}_p^-}^2\bigr)
+
\frac{1}{2}\bigl(\norm{\tilde{\sigma}_p^+}^2-\norm{\tilde{\sigma}_o^-}^2\bigr)\\
&-\frac{1}{2}\norm{\tilde a_u-\tilde{\sigma}_p^-}^2
-\frac{1}{2}\norm{\tilde a_\omega-\tilde{\sigma}_o^-}^2.
\end{aligned}
\]
Rearranging the first two differences and using the scattering identities
\(\frac{1}{2}(\norm{\tilde{\sigma}_p^+}^2-\norm{\tilde{\sigma}_p^-}^2)=-(G^\top\tilde\omega)^\top\tilde p\) and
\(\frac{1}{2}(\norm{\tilde{\sigma}_o^+}^2-\norm{\tilde{\sigma}_o^-}^2)=\tilde u^\top y\)
gives \eqref{eq:filtered-wave-channel-passivity}.
\qed
\end{proof}
\begin{remark}
\label{rem:filtered-wave-transient}
In the sampled RBC
implementation, blockwise optimizer updates may generate
fast variations in \(\tilde u\), and therefore in the outgoing wave
\(\tilde{\sigma}_o^+\). The filters \eqref{eq:filtered-wave-channel} dissipate such variations through the mismatch
terms \(\norm{a_u-\tilde{\sigma}_p^-}^2\) and \(\norm{a_\omega-\tilde{\sigma}_o^-}^2\). Meanwhile, their
unit-DC-gain property preserves the steady-state wave
interconnection. This
explains why wave-domain interface filtering can improve large-delay transient
behavior without changing the target equilibrium.
\end{remark}
\section{Verification}
\label{sec:verification}

To verify the proposed plant-optimizer interconnection and its RBC digital
realization, we use the standard IEEE 14-bus transmission-network benchmark \cite{zimmermanMATPOWERSteadyStateOperations2011,matpower_manual}.
The network model is built on the \(100\)~MVA base. The controllable generation
vector contains the four non-slack units at buses \(2\), \(3\), \(6\), and
\(8\). The area of interest is chosen as buses \(1\)-\(5\), so the scheduled
transfer is the signed aggregate flow across the corridors \(4\text{-}7\),
\(4\text{-}9\), and \(5\text{-}6\). The quadratic dispatch cost is
\(J(u)=\frac{1}{2}(u-u^{\mathrm{ref}})^\top Q(u-u^{\mathrm{ref}})\)
with \(Q=\diag(3,5,6,7)\), and the pre-disturbance setpoint is
\(u^{\mathrm{ref}}=[40,\ 0,\ 0,\ 0]^\top\ \mathrm{MW}\).

A \(6\)~MW load increase is applied at \(t=5\)~s, split as \(60\%\) at bus~4
and \(40\%\) at bus~5. The tie-line schedule is fixed at its pre-disturbance
value \(87.7\)~MW. To activate a nontrivial congestion channel, the upper limit
of line \(2\text{-}4\) is tightened to \(55.65\)~MW, while the remaining line
limits are placed \(80\)~MW above their base flows. The RBC run uses
sampling period \(\varepsilon=0.0006\)~s and a uniform block activation over
the \(39\) cyber blocks induced by the centralized state partition.

Table~\ref{tab:rbc-computation} reports the optimizer workload: the RBC
implementation updates one block and \(1.87\) cyber coordinates per step on
average, amounting to \(2.56\%\) of the full sampled-update coordinate count
over the \(300\)~s horizon.

\begin{table}[t]
\centering
\caption{Cyber-state update counts over the \(300\)~s verification horizon at
\(\varepsilon=0.0006\)~s.}
\label{tab:rbc-computation}
\small
\resizebox{\linewidth}{!}{%
\begin{tabular}{lcccc}
\hline
Scheme & Blocks/step & Coord./step & Total updates & Relative load \\
\hline
Full sampled update & 39 & 73 & $3.65\times10^7$ & 100\% \\
RBC & 1 & 1.87 & $9.36\times10^5$ & 2.56\% \\
\hline
\end{tabular}
\par}
\end{table}

We consider two delayed plant-optimizer channels. The unfiltered scattering
interconnection uses impedance \(\eta=1.0\) and \(11\)~ms delay in each
direction, giving \(22\)~ms round-trip delay. The filtered case adds
first-order wave-variable filters with unit DC gain and time constants
\(10\)~ms and \(20\)~ms in the control and measurement directions,
respectively, and uses \(40\)~ms delay in each direction, giving \(80\)~ms
round-trip delay. In both cases, solving the post-disturbance steady-state
optimization problem yields the dispatch
\(u^\star=[38.52,\ 7.48,\ 0,\ 0]^\top\ \mathrm{MW}\), restores the tie-line
schedule to \(87.7\)~MW, and activates the upper line constraint on line
\(2\text{-}4\).

\begin{figure*}[t]
\centering
\begin{minipage}[t]{0.49\textwidth}
\centering
\includegraphics[width=\linewidth]{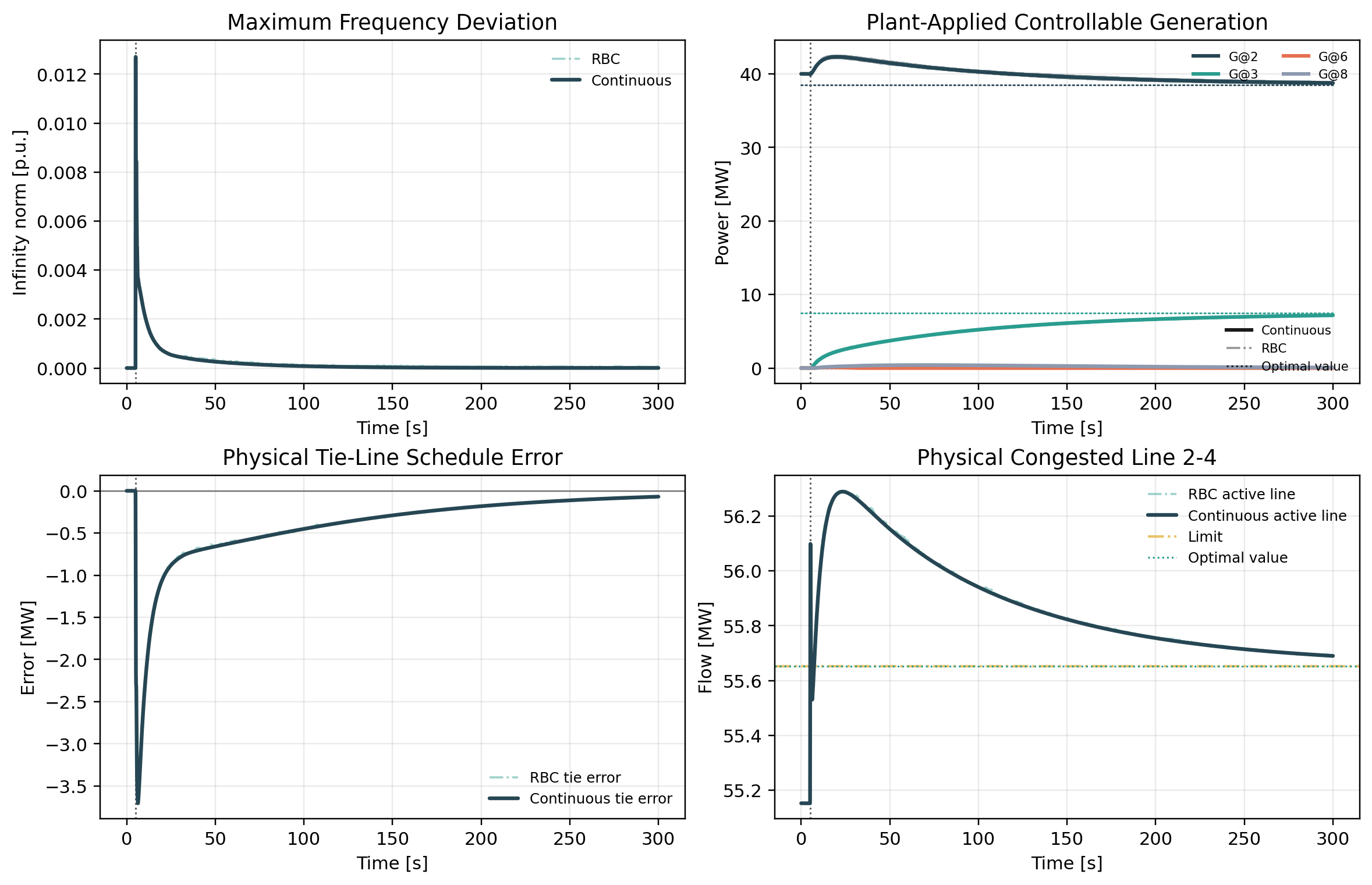}
\\[2pt]
\small (a) unfiltered, \(22\) ms round-trip delay
\end{minipage}\hfill
\begin{minipage}[t]{0.49\textwidth}
\centering
\includegraphics[width=\linewidth]{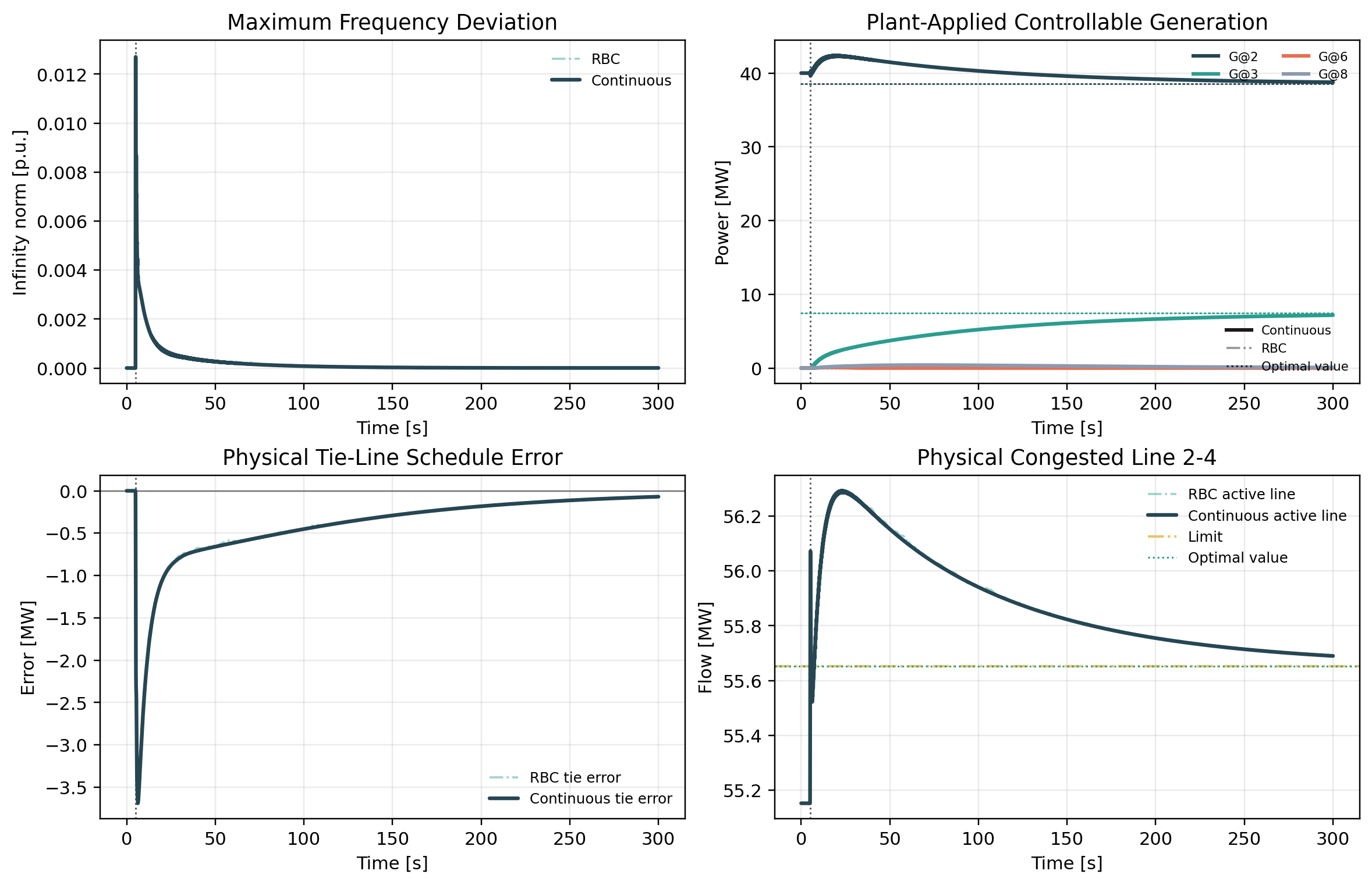}
\\[2pt]
\small (b) filtered, \(80\) ms round-trip delay
\end{minipage}
\caption{Verification of the IEEE 14-bus transmission-network benchmark.
Panel (a) verifies the unfiltered delayed scattering interconnection with
\(11\)~ms delay in each direction. Panel (b) verifies the wave-domain filtered
interface with \(40\)~ms delay in each direction and first-order filters with
\(10\)~ms and \(20\)~ms time constants. Both panels compare the continuous-time
interconnection and the RBC digital implementation.}
\label{fig:verification_ieee14_combined}
\end{figure*}

Fig.~\ref{fig:verification_ieee14_combined} shows that the RBC
implementation tracks the continuous-time delayed interconnection in both
cases, restores frequency and the scheduled tie-line exchange, and respects
the active line constraint at the constrained steady state. The wave-domain
filters allow the closed-loop system to maintain comparable performance under
a larger communication delay, while the RBC implementation achieves these
behaviors with much lower per-update cyber computation.
\section{Conclusion}
\label{sec:conclusion}
This paper developed a centralized secondary frequency-control framework that
simultaneously enforces tie-line schedule restoration, transmission congestion
constraints, global power balance, and generation limits under bidirectional
communication delays. The delayed plant-optimizer interconnection was modeled
through passive scattering channels, and wave-domain filtering was introduced
to add dissipation without changing the steady-state interconnection. To reduce
the central computational burden, we further developed an RBC sampled-data
implementation of the projected primal-dual controller. We showed that this
implementation preserves the target delayed equilibrium and achieves local
mean-square geometric stability under suitable regularity and step-size
conditions. Verification on the IEEE 14-bus system confirmed the predicted
closed-loop behavior and the computational reduction delivered by the RBC
implementation.
\bibliographystyle{unsrt}
\bibliography{references}

\end{document}